\documentclass[12pt]{article}
\usepackage{amsmath}
\usepackage{amsthm}
\usepackage{graphicx}
\usepackage{enumerate}% http://ctan.org/pkg/enumerate
\usepackage[colorinlistoftodos]{todonotes}
\usepackage[a4paper, margin=1in]{geometry}
\usepackage{amssymb}
\usepackage{bm}
\usepackage{natbib}
\usepackage{times}
\usepackage{caption}
\usepackage{subcaption}

\usepackage[plain,noend]{algorithm2e}
\usepackage{float}
\newtheorem{theorem}{Theorem}
\newtheorem{lemma}{Lemma}
\newtheorem{remark}{Remark}
\allowdisplaybreaks
\usepackage{fancyhdr}
\fancypagestyle{plain}{%
    \fancyhf{}%
    \fancyfoot[R]{\thepage}%
}

\usepackage{color,soul}
\usepackage{array}
\newcolumntype{L}[1]{>{\centering}m{#1}}
\begin{document}
\title{Empirical Likelihood for Change Point Detection in Autoregressive Models}
%% The left and right page headers are defined here:
\author{Ramadha D. Piyadi Gamage$^1$, Wei Ning$^2$\footnote{Corresponding author. Email:
wning@bgsu.edu} \\$^1$Department of Mathematics\\Western Washington University, Bellingham, WA 98225, USA\\
$^2$Department of Mathematics and Statistics\\Bowling Green State University, Bowling
Green, OH 43403, USA}
%\author{Ramadha D. Piyadi Gamage \\ Department of
%Mathematics and Statistics\\Bowling Green State University, Bowling
%Green, OH 43403, USA}
\date{}
\maketitle

\begin{abstract}
\noindent Change point analysis has become an important research topic in many fields of applications. Several research work has been carried out to detect changes and its locations in time series data. In this paper, a nonparametric method based on the empirical likelihood is proposed to detect the structural changes of the parameters in autoregressive (AR) models . Under certain conditions, the asymptotic null distribution of the empirical likelihood ratio test statistic is proved to be the extreme value distribution. Further, the consistency of the test statistic has been proved.  Simulations have been carried out to show that the power of the proposed test statistic is significant. The proposed method is applied to real world data set to further illustrate the testing procedure. 
\end{abstract}

\section{Introduction}
\noindent Change point analysis introduced by Page (1954, 1955) has become popular due to its usage in wide variety of fields, such as stock market analysis, quality control, traffic mortality rate, geology data analysis, genetics, etc. It concerns both detecting whether or not a change(s) has (have) occurred, and identifying the location(s) of any such change(s). Several methods to identify and estimate the change points in the change point problem are proposed by scholars. Bayesian approach to detect changes in the mean has been discussed by Chernoff and Zacks (1964) and Sen Srivastava (1975). Further, Cs\"org\'o and Horv\'ath (1997) and Chen and Gupta (2000) established asymptotic results on parametric change point models.  Hawkins (1977), Worsley (1986) and Gombay and Horv\'ath (1994) are a few among the many researchers who discussed change point problem  under the parametric settings. However, the parametric methods are no longer applicable if the underlying distribution is completely unknown. In such a case, a nonparametric approach should be considered as an alternative. One such popular nonparametric approach is the Cumulative Sum (CUSUM) method. Most authors have assumed that the observations are independent and studied the case where two distributions differ only in location. Combining nonparametric approaches along with the change point detection has been studied by many scholars over the past years. Aue and Horv\'ath (2012) discussed two methods, namely, Cumulative Sum (CUSUM) and Likelihood Ratio Test (LRT), on how they can be modified for data exhibiting serial dependence. Further, they provided some insight to the sequential procedure as well. Lee et. al. (2003) also discussed about the Cusum test for changes of parameters in time series models and considered the changes of the parameters in a random coefficient autoregressive model AR(1) and that of the autocovariances of a linear process.

The change point problem may be viewed as a two-sample test adjusted for the unknown break location, thus leading to max-type procedures. Correspondingly, asymptotic relationships are derived to obtain critical values for the tests. In general, the change point problem can described as follows. Let $ x_1,x_2,...,x_n $ be a sequence of independent random vectors (variables) with probability distribution functions $ F_1,F_2,...,F_n $, respectively. More specifically, suppose that the distributions $ F_1,F_2,...,F_n $ belong to a common parametric family F($\theta$), where $\theta \in R^p$, then the change point problem is to test the hypotheses about the population parameters $\theta_i , i=1,...,n $ $$H_0:\theta_1=\theta_2=...=\theta_n=\theta (unknown),$$ versus the alternative $$H_1:\theta_1=...=\theta_{k_1} \neq \theta_{k_1+1}=...=\theta_{k_2}\neq...\neq\theta_{k_{q-1}}=...=\theta_{k_q}\neq\theta_{k_q+1}...=\theta_n ,$$ where $q$ and $k_1,k_2,...,k_q$ are unknown and need to be estimated.

Empirical likelihood introduced by Owen (1988, 1990) is one of the popular and powerful nonparametric approaches. It has been widely used due to the robustness of its nonparametric nature and the efficiency of its likelihood construction. Kolaczyk (1994) used empirical likelihood with generalized linear models. Further, Qin and Lawless (1994) obtained estimating equations and derived asymptotic properties of the test statistic. Many scholars have discussed about the empirical likelihood ratio test for a change point in linear models, such as Zou et al. (2007) \nocite{zou2007empirical}, Liu et al. (2008) \nocite{liu2008empirical}, and Ning (2012). Since the empirical likelihood was originally proposed for independent data, it is difficult to apply it to dependent data such as time series data. Several approaches suggested to reduce the dependent data problem into an independent data problem. Owen (2001) suggested using the conditional likelihood to remove the dependence structure and generate the estimating equations. Kitamura (1997) used block-wise empirical likelihood method which preserves the dependence of data, and the resulting likelihood ratios have been used to construct asymptotically valid confidence intervals. Ogata (2005) and Nordman and Lahiri (2006) independently formulated a frequency domain empirical likelihood (FDEL) using spectral estimating equations which can be used for short- and long- range dependent data. Bai and Perron (1998) proposed CUSUM and F-based statistics for change point detection. Baragona et al. (2013) compared it with the test they proposed for change point detection based on the empirical likelihood approach for change point detection.

To deal with the situation of multiple changes, it traditionally uses the binary segmentation method proposed by Vostrikova (1981). The advantage of using this method is that it detects number of change points and estimates their locations simultaneously as well as the consistency of this method has been established. Hence, the general hypothesis of the change point problem can be simplified as the hypothesis of no change point versus a single change point, i.e. the alternative hypothesis is:
$$H_1:\theta_1=...=\theta_k\neq\theta_{k+1}=...=\theta_n,$$
where $k$ is the location of the single change point at this stage. If $H_0$ is not rejected, then the process is stopped and we conclude that there is no change. If $H_0$ is rejected, then there is a change point and the two subsequences before and after the change point found are tested for a change. This process is repeated until there are no subsequences having change points.

In this paper, we propose a test statistic based on the empirical likelihood approach for detecting changes in a time series model. In Section 2, the change point problem in time series models has been introduced for AR(p) model. The empirical likelihood procedure for change point detection is described in Section 3. The null asymptotic distribution of the test statistic and the consistency of the test along with the proofs are provided under Section 4. Simulations are carried out in Section 5 and a real data application is given in Section 6. Section 7 provides some discussion and proofs of results are given in the Appendix.

\section{Changepoint Problem in AR(p) Model}
\noindent Consider the stationary \textsc{ar}(p) model with the mean 0. 
$$X_t =
\begin{cases}
\sum_{i=1}^p \phi_i X_{t-i} + \epsilon_t; 1 \leq t \leq k \\
\sum_{i=1}^p \phi_i^* X_{t-i} + \epsilon_t; k+1 \leq t \leq n,
\end{cases}$$
where $\epsilon_t$'s are independent random variables with mean zero and variance $\sigma^2$, (i.e. White noise process),$\phi_1, \phi_2, ..., \phi_p, \phi_1^*, \phi_2^*, ..., \phi_p^*$ are all unknown parameters, and $k$ is the unknown change location which needs to be estimated. Denote $\boldsymbol{\delta}=\Phi^*-\Phi$, where $\Phi = (\phi_1, \phi_2, ..., \phi_p)^\prime$ and $\Phi^* = (\phi_1^*, \phi_2^*, ..., \phi_p^*)^\prime$. Therefore, the change point problem is to test the null hypothesis of no change in the autoregressive parameters versus the alternative hypothesis of one unknown change, i.e.,
$$H_0:\boldsymbol{\delta} = 0 \quad vs \quad H_1:\boldsymbol{\delta} \neq 0, \text{\small at least one $\delta_i \neq 0$}.$$
Hence, under the alternative hypothesis, there is a change in at least one of the $p$ parameters at an unknown location. We denote $\beta_0 = (\Phi, 0)^\prime \in \mathbb{R}^{2p}$ and $\beta_1 = (\Phi, \boldsymbol{\delta})^\prime\in \mathbb{R}^{2p}$ to be the parameter vectors under the null and the alternative hypothesis respectively. According to Owen (1991), we derive the estimating functions to be
\begin{align}
g_1(X_i, \beta_0) = (X_i, X_{i-1} \epsilon_i,..., X_{i-p} \epsilon_i, {\epsilon_i}^2 - \sigma^2),
\end{align}
where $\epsilon_i = X_i -  \sum_{r=1}^p \phi_i X_{i-r},$ and
\begin{align}
g_2(X_j, \beta_1) = (X_j, X_{j-1} \epsilon_j,..., X_{j-p} \epsilon_j, {\epsilon_j}^2 - \sigma^2),
\end{align}
where $\epsilon_j =  X_j - \sum_{r=1}^p \phi_j^* X_{j-r}.$ It is easy to see that
\begin{align*}
&E[g_1(X_i,\beta_0)]=0, \\
&E[g_2(X_j,\beta_1)]=0.
\end{align*}for every $1 \leq i \leq k$ and $k+1 \leq j \leq n$.

\section{Empirical Likelihood for AR(p) Changepoint Model}
WLOG, we assume one change point at an unknown location $k$. Let
\begin{align}
\Omega_{H_0} = \big\{ (p,q,\beta_0) | \sum_i p_i g_1(X_i,\beta_0)=\sum_j q_j g_2(X_j, \beta_0) = 0 \big\},
\end{align}
and
\begin{align}
\Omega_{H_1} = \big\{ (p,q,\beta_1) | \sum_i p_i g_1(X_i,\beta_0)=0, \sum_j q_j g_2(X_j, \beta_1) = 0 \big\}
\end{align}
be the parameter spaces under $H_0$ and $H_1$, respectively, where $p=(p_1, p_2, ..., p_k)$ and $q = (q_{k+1}, q_{k+2},..., q_n)$ are the probability vectors such that $\sum_{i=1}^{k}p_i = 1$, $\sum_{j=k+1}^{n}q_j = 1$ and $p_i \geq 0$, $q_j \geq 0$. If a change occurs at $k$, then the empirical likelihood ratio test statistic is defined as,
\begin{align*}
-2 \log \Lambda_k &= -2 \log \frac{\underset{H_0}{\sup} \left \{ \prod_{i} p_i \prod_{j} q_j | (p,q,\beta_0) \in \Omega_{H_0} \right \}}{\underset{H_1}{\sup} \left \{ \prod_{i} p_i \prod_{j} q_j | (p,q,\beta_1) \in \Omega_{H_1} \right \}}\\
&= Z_{H_0,k} - Z_{H_1,k},
\end{align*}
where
$$Z_{H_0,k} = -2 \underset{H_0}{\sup} \left \{ \sum_i \log kp_i + \sum_j \log (n-k) q_j | (p,q,\beta_0) \in \Omega_{H_0} \right \}, $$
$$Z_{H_1,k} = -2 \underset{H_1}{\sup} \left \{ \sum_i \log kp_i + \sum_j \log (n-k) q_j | (p,q,\beta_1) \in \Omega_{H_1} \right \}. $$
The null hypothesis is rejected for a sufficiently large value of $\underset{1<k<n}{\max}-2 \log \Lambda_k$. Let $\theta_{nk} = \frac{k}{n}$. A Lagrangian argument gives,
$$p_i = \frac{1}{n\theta_{nk} ( 1 + \theta_{nk}^{-1}\lambda_1^\prime g_1(X_i, \cdot))} $$
and
$$q_j = \frac{1}{n(1-\theta_{nk}) ( 1 + {(1-\theta_{nk})}^{-1}\lambda_2^\prime g_2(X_j, \cdot))} $$
where $\lambda_1$ and $\lambda_2$ are chosen such that $\sum_i p_i g_1(X_i,\cdot)= 0$ and $\sum_j q_j g_2(X_j, \cdot) = 0$. Therefore,
under $H_0,$ we obtain
$$Z_{H_0,k} = 2 \underset{\beta_0}{\inf} \underset{\lambda_1, \lambda_2}{\sup}\left\{ \sum_i \log(1 + \theta_{nk}^{-1} \lambda_1^\prime g_1(X_i,\beta_0)) + \sum_j \log(1 + {(1 - \theta_{nk})}^{-1} \lambda_2^\prime g_2(X_i,\beta_0)) \right\}.$$
Let $\lambda = (\lambda_1^\prime, \lambda_2^\prime)^\prime$. The score functions are defined as:
$$Q_{1n}(\beta_0,\lambda) = \frac{\partial Z_{H_0,k}}{\partial \lambda} = \frac{1}{n} \sum_m \frac{1}{1 + \theta_m^{-1} \lambda^\prime g(X_m,\beta_0)} \theta_m^{-1} g(X_m,\beta_0),$$
and
$$Q_{2n}(\beta_0,\lambda) = \frac{\partial Z_{H_0,k}}{\partial \beta_0} = \frac{1}{n} \sum_m \frac{1}{1 + \theta_m^{-1} \lambda^\prime g(X_m,\beta_0)} \theta_m^{-1} \bigg(\frac{\partial g(X_m, \beta_0)}{\partial \beta_0}\bigg)^\prime g(X_m,\beta_0), $$
where
\begin{align*}
&\theta_m^{-1} = \theta_{nk}^{-1} \mathbf{1}_{\{1\leq m \leq k\}} + {(1-\theta_{nk})}^{-1} \mathbf{1}_{\{k+1 \leq m \leq n\}},\\
&g(X_m, \beta_0) = g_1(X_m,\beta_0) \mathbf{1}_{\{1\leq m \leq k\}} + g_2(X_m, \beta_0) \mathbf{1}_{\{k+1 \leq m \leq n\}}.
\end{align*}
Under certain regularity conditions, Qin and Lawless (1994) showed, there exists ($\tilde{\beta_0},\tilde{\lambda}$) such that,
$$Q_{1n}(\tilde{\beta_0},\tilde{\lambda}) = 0 \text{ and } Q_{2n}(\tilde{\beta_0},\tilde{\lambda}) = 0.$$
Hence, we obtain $Z_{H_0,k} = 2l_E(\tilde{\Phi}^\circ, \tilde{\mu}^\circ, 0)$.

Similarly, under $H_1$ we have, $Z_{H_1,k} = 2l_E(\tilde{\Phi}, \tilde{\mu}, \delta)$. Then the empirical likelihood ratio statistic can be rewritten as
\begin{align}
-2 \log \Lambda_k = 2l_E(\tilde{\Phi}^\circ, \tilde{\mu}^\circ, 0) - 2l_E(\tilde{\Phi}, \tilde{\mu}, \delta).
\end{align}
Since $k$ is unknown, $H_0$ is rejected when the maximally selected log-likelihood ratio statistic,
$$Z_n = \underset{\theta_{nk} \in \Theta_n}{\max} \{-2 \log \Lambda_k\},$$
where $\Theta_n = \{k/n: k=1,2,...,n\}$, is sufficiently large.

When $k$ or $n-k$ is too small, then the minimax estimators of empirical likelihood $(\tilde{\beta}_1, \tilde{\lambda})$ may not exist. Hence we consider the trimmed likelihood ratio statistic where the range of $k$ is selected arbitrarily as follows. The Trimmed likelihood ratio statistic is defined as,
\begin{equation}
Z_n^* = \underset{\theta_{nk} \in \Theta_n^*}{\max} \{-2 \log \Lambda_k\}, %\tag{3.1}%\label{eq:3.1}
\end{equation}
where $\Theta_{nk}^* =\{k/n :k = n_{T_1}, n_{T_1}+1,..., n - n_{T_2}\}$. According to Perron and Vogelsang (1992), the selection of $n_{T_1}$ and $n_{T_2}$ can be arbitrary. In our work, we choose $n_{T_1} = n_{T_2} = 2 [n^{\frac{1}{2}}]$, where $[x]$ means the largest integer not larger than $x$. If $H_0$ is true, then $Z_n^*$ follows an asymptotic extreme value limit distribution.  The convergence to the extreme value limit can be slow and asymptotic test often tends to be too conservative in finite samples.

\section{Main Results}
The results are similar to the ones by Cs\"{o}rg\'{o} and Horv\'{a}th (1997). Under mild regularity conditions, the following theorem holds.
\begin{theorem}
Let $\beta^*$ be the true parameter. Suppose that $E{||g(X,\beta^*)||}^3 < \infty$, $E||\epsilon||^4 < \infty$, and $E(\epsilon \epsilon^\prime)$ is positive definite. If $H_0$ is true, then we have
$$\underset{n\rightarrow \infty}{\lim} Pr\{A(\log u(n)) {(Z_n^*)}^{\frac{1}{2}} \leq t + D_r(\log u(n))\} = \exp(-e^{-t})$$
for all t, where $A(x)=(2\log x)^{\frac{1}{2}}$, $D_r(x)=2 \log x + (r/2) \log \log x - \log \Gamma (r/2)$, $u(n) = \frac{n^2 + (2 \lfloor n^{\frac{1}{2}} \rfloor)^2 - 2n \lfloor n^{\frac{1}{2}} \rfloor}{(2 \lfloor n^{\frac{1}{2}} \rfloor)^2}$, and $r$ is the dimension of the parameter $\boldsymbol{\delta}$.
\end{theorem}

\begin{theorem}
Under the conditions of Theorem 1 and the condition that for every fixed parameter $\delta=\beta^*-\beta\neq0$, there exists a positive constant $c_0>0$ satisfy that $\infty > \smash{\displaystyle\inf_{\delta\neq 0}} \smash{\displaystyle\sup_{\lambda}} E\log \big[ 1 + \lambda^\prime x(x^\prime \delta + e) \big] \geq c_0 >0$, holds. If $H_1$ is true, assume that $\theta_{k_0} = \frac{k_0}{n} \rightarrow \theta_0\in (0,1)$ as $n\rightarrow\infty$, then ELR test statistic is consistent, i.e. there exists a constant $c>0$ such that
$$P(Z_n >cn) \rightarrow 1.$$
\end{theorem}

\begin{theorem}
Under the conditions of Theorem 1 and the condition that for every fixed parameter $\delta=\beta^*-\beta\neq0$, there exists a positive constant $c_0>0$ satisfy that $\infty > \smash{\displaystyle\inf_{\delta\neq 0}} \smash{\displaystyle\sup_{\lambda}} E\log \big[ 1 + \lambda^\prime x(x^\prime \delta + e) \big] \geq c_0 >0$, holds. If $H_1$ is true, assume that $\theta_{k_0}= \frac{k_0}{n} \rightarrow \theta_0 \in (0,1)$ as $n\rightarrow \infty$, we have $\hat{\theta}_k \rightarrow\theta_0$ in probability as $n\rightarrow \infty$.
\end{theorem}
Proofs are given in the Appendix.

\section{Simulation Study}
A Monte Carlo simulation has been conducted to illustrate the performance of the proposed method. Consider the following AR(1) model with mean $0$:
$$X_t =
\begin{cases}
0.1 X_{t-1} + \epsilon_t; 1 \leq t \leq k \\
0.5 X_{t-1} + \epsilon_t; k+1 \leq t \leq n,
\end{cases}$$
where $\epsilon_t$ is the white noise with mean zero and variance $\sigma^2$. Four different distributions are considered for $\epsilon_t$: (i) $N(0,1)$, (ii) $\exp(1)-1$, (iii) $\frac{1}{2 \sqrt{2}}(\chi_4^2 -4)$, and (iv) $\frac{1}{\sqrt{2}}t_4$. The power of the proposed test in detecting changes in parameters of the AR(1) model has been calculated for two different sample sizes: n=100, 150 and 250. Different change locations have been considered under each sample size. Additional simulations have been carried out to compute the empirical critical values under different significance levels which are turned out to be close to the theoretical critical values for the corresponding significance level. Hence, we use the theoretical critical value 2.9702 with~$\alpha=0.05$~for power calculations with 1000 simulations. The results are listed in Table 1. It can be seen that the power of the hypothesis test of AR(1) model increases with the sample size. The power values under a given change location are approximately similar for the four different error distributions. This maybe due to the fact that the three distributions $\exp(1)-1, \frac{1}{2 \sqrt{2}}(\chi_4^2 -4) \text{ and } \frac{1}{\sqrt{2}}t_4$ are standardized. When the change location is farther away from the starting location, then the power tends to decrease. Intuitively, this maybe due to the dependency existing in the data set.

\begin{table}[H]
\caption{\textit{Power of the hypothesis test for \textsc{AR(1)} model}} % title of Table
\centering
\small
\begin{tabular}{c c c c c}
\hline\hline %inserts double horizontal lines
k & $N(0,1)$ & $\exp(1)-1$ & $\frac{1}{2 \sqrt{2}}(\chi_4^2 -4)$ & $\frac{1}{\sqrt{2}}t_4$ \\ [0.5ex]
\hline
\multicolumn{5}{c}{n = 100} \\ [0.5ex]
20    & 0.802 & 0.816 & 0.815 & 0.808 \\
30    & 0.765 & 0.775 & 0.774 & 0.779 \\
40    & 0.723 & 0.732 & 0.754 & 0.747 \\
50    & 0.656 & 0.669 & 0.627 & 0.674 \\
80    & 0.296 & 0.292 & 0.283 & 0.331 \\ [1ex]

\multicolumn{5}{c}{n = 150} \\ [0.5ex]
30    & 0.929 & 0.929 & 0.924 & 0.903 \\
45    & 0.913 & 0.901 & 0.893 & 0.872 \\
60    & 0.862 & 0.844 & 0.853 & 0.845 \\
75    & 0.806 & 0.799 & 0.813 & 0.804 \\
120    & 0.385 & 0.397 & 0.426 & 0.450 \\ [1ex]

\multicolumn{5}{c}{n = 250} \\ [0.5ex]
50    & 0.993 & 0.988 & 0.990 & 0.976 \\
80    & 0.983 & 0.981 & 0.976 & 0.965 \\
100   & 0.966 & 0.966 & 0.976 & 0.948 \\
125   & 0.941 & 0.927 & 0.928 & 0.926 \\
200   & 0.621 & 0.603 & 0.626 & 0.622 \\
\end{tabular}
\label{table:Table1}
\end{table}

\section{Application}
In this section, we study the data which consists of monthly average soybean prices achieved by farmers in Illinois from January 1960 to November 2008 with the sample size 587. The prices are given in dollars per bushel. This data was analyzed by Balcombe et al. (2007) who considered the threshold AR(1) models  for modeling the prices of agricultural products. Berkes et al. (2011) studied this data set by proposing the likelihood ratio test to detect the structural change of an AR model to threshold AR model.  We apply the proposed EL method for AR(1) changepoint model to detect the structural change in the same data set. Figure 1 shows the time series plot for the given data.
\begin{figure}[h]
\centering
\caption{Time Series for the Monthly average soybean prices}
\includegraphics[scale=0.5]{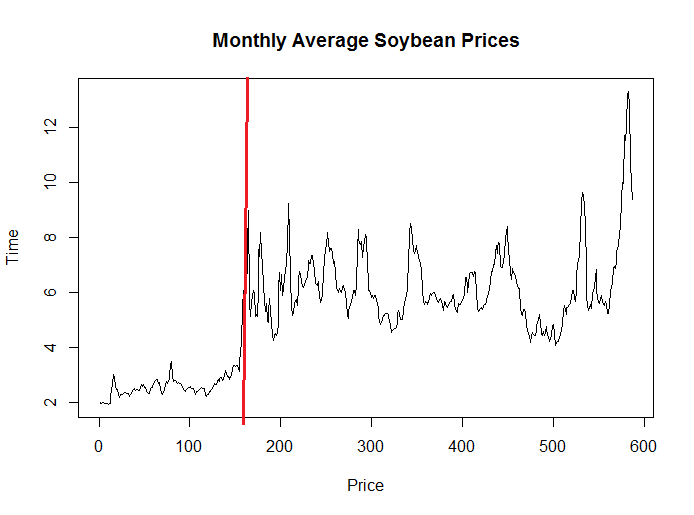}
\end{figure}
In order to test if there are significant changes, we use the $Z_n^*$ from (6). The value of $Z_n^*$ is 16.07426. Using the critical values derived under the Theorem are given in Table 2, we have sufficient evidence to reject the null hypothesis that there is no change.
\begin{table}[h]
\centering
\caption{Theoretical Critical values}
\begin{tabular}{c|c|c}
\hline
\hline
$\alpha =0.01$ & $\alpha=0.05$ & $\alpha=0.10$ \\ \hline
4.600149     & 2.970195    & 2.250367    \\ \hline
\end{tabular}
\end{table}

\section{Discussion}
In this paper, we discuss developing an EL-based detecting procedure for structural changes in time series data, i.e. testing null hypothesis of no change versus alternative hypothesis of one change. A test statistic is derived for a fixed change location and the max-type of test statistic over all possible change locations is considered. The asymptotic null distribution of the test statistic has been established as extreme value distribution. Simulations to compute the power in AR(1) model have been carried out with different sample sizes and different error distributions in order to illustrate the performance of the proposed test statistic. The results indicate that the proposed method is efficiently identify the changes in a given time series data set. We should point out that, due to the slow convergence of the proposed test statistic in Theorem 1, the moderate or the large sample size is recommended to achieve the good approximation (See Cs\"{o}rg\H{o} and Horv\'{a}th, 1997). If the sample size is small, the bootstrap is suggested to obtain the approximated p-values in practice. 

As for future work, we plan to extend the proposed method to other stationary time series models such as MA, ARMA, GARCH models along with corresponding analytic results and simulations. Comparisons to other existing methods will be done. Further, sequential change point detection based on EL method is to be studied where the sample size is a random variable and the null hypothesis of sequential structural stability will be rejected as soon as a change is detected. Hence, the objective in sequential change point detection is to detect such a change with a minimum number of false alarms. A nonparametric testing procedure based on EL method will be proposed and related asymptotic results will be studied.

\section*{Appendix}
In order to prove Theorem 1, we need following Lemmas.
\begin{lemma}
Assume that for $i=1,2, E[g_i(X, \beta_0)g_i^\prime(X, \beta_0)]$ is positive definite, $\frac{\partial g_i(X, \beta)}{\partial \beta}$ is continuous in a neighborhood of the true value $\beta_0$, $E \Big[ \Big(\frac{\partial g_i(X, \beta_0)}{\partial \beta^\prime}\Big) \Big(\frac{\partial g_i(X, \beta_0)}{\partial \beta^\prime}\Big)^\prime \Big]$, $E \Big[\frac{\partial^2 g_i(X, \beta)}{\partial \beta \partial \beta^\prime}\Big]$, $E\Big[ \Big(\frac{\partial g_i(X, \beta)}{\partial \beta^\prime}\Big)^\prime g_i(X, \beta) \Big]$ and $E\parallel g_i(X, \beta) \parallel ^3$ are all bounded in the neighborhood of the true value $\beta_0$. Then, as $n \rightarrow \infty$, $\exists \tilde{\beta}$, $\tilde{\lambda}=\lambda(\tilde{\beta})$ with probability 1 satisfying,
$$Q_{1n}(\tilde{\beta}, \tilde{\lambda}) = 0, \quad Q_{2n}(\tilde{\beta}, \tilde{\lambda}) = 0 \text{ and } \parallel \tilde{\beta}-\beta_0 \parallel = O_p(m^{-\frac{1}{2}}),$$
where
\begin{align*}
Q_{1n}(\beta, \lambda)&=\sum_l \frac{1}{1 + \lambda^\prime(\beta) \theta_l^{-1}g(x_l, \beta)}\theta_l^{-1}g(x_l, \beta),\\
Q_{2n}(\beta,\lambda)&=\sum_l \frac{1}{1 + \lambda^\prime(\beta) \theta_l^{-1}g(x_l, \beta)} \theta_l^{-1} \bigg(\frac{\partial g(x_l, \beta)}{\partial \beta}\bigg)^\prime \lambda(\beta).
\end{align*}
\end{lemma}

\begin{proof}
First we will show
\begin{align*}
\lambda(\beta) &= \epsilon_k O_p(m^{-\frac{1}{2}}) \\
               &= \Big[ \frac{1}{n} \sum_{l=1}^n \theta_l^{-2} g(x_l, \beta) g^\prime(x_l, \beta) \Big]^{-1} \Big[ \frac{1}{n} \sum_{l=1}^n  \theta_l^{-1} g(x_l, \beta)\Big] + \epsilon_k o_p(m^{-\frac{1}{2}}),
\end{align*}
where $\epsilon_k = \min \{\theta_k, 1- \theta_k\}$ and $m=n\epsilon_k = \min\{k, n-k\}$.\\
Let $\beta-\beta_0 = um^{-\frac{1}{2}}$ for $\beta \in \{\beta:\parallel\beta-\beta_0\parallel=m^{-\frac{1}{2}}\}$ where $\parallel u\parallel=1$. Let~$\lambda$~be the solution of the function~$f(\lambda)$~given by the first score function defined in Section 3.
\begin{equation}
f(\lambda) = \frac{1}{n} \sum_{l=1}^n \frac{\theta_l^{-1}}{1 + \lambda^\prime(\beta) \theta_l^{-1} g(x_l,\beta)} g(x_l,\beta) = 0.
\tag{A.1}\label{eq:A1}
\end{equation}
Let $\lambda=\rho u$ where $u=(\beta-\beta_0)m^\frac{1}{2}$ and $\parallel u\parallel=1$.
\begin{align*}
0&=\parallel f(\rho u)\parallel \\
&\geq |u^\prime f(\rho u)| \\
&= \frac{1}{n} \Big|u^\prime \Big(\sum_l \theta_l^{-1} g(x_l,\beta) - \rho \sum_l \frac{\theta_l^{-2} g(x_l, \beta)u^\prime g(x_l, \beta)}{1 + \rho u^\prime \theta_l^{-1} g(x_l, \beta)} \Big) \Big| \\
&\geq \frac{\rho}{n} u^\prime \sum_l \frac{\theta_l^{-2} g(x_l, \beta)u^\prime g(x_l, \beta)}{1 + \rho u^\prime \theta_l^{-1}} u - \frac{1}{n} \Big| \sum_{j=1}^p e_j \sum_l \theta_l^{-1}g(x_l, \beta)\Big| \tag{\text{ where $e_j$ is the unit vector in the $j^{th}$ coordinate direction.}}\\
&\geq \frac{\rho u^\prime Su}{1 + \rho \theta_l g^*} - O_p(m^{-\frac{1}{2}}),\tag{\text{where $g^*=\smash{\displaystyle\max_{l}}  g(x_l, \beta)$ and $S=\frac{1}{n}\sum_l \theta_l^{-2} g(x_l,\beta) g^\prime(x_l,\beta).$}}
\end{align*}
Since $u^\prime Su \geq \sigma_p + o_p(1)$, where $\sigma_p>0$ is the smallest eigen value of $\Sigma$, then
$$\frac{\rho}{1 + \rho \theta_l g^*} = O_p(m^{-\frac{1}{2}})$$
So, $\parallel \lambda \parallel = \rho =  O_p(m^{-\frac{1}{2}})$. \\
Let $\gamma_l = \lambda^\prime(\beta)\theta_l^{-1} g(x_l, \beta)$. Then, $\smash{\displaystyle\max_{l}} |\gamma_l| = O_p(m^{-\frac{1}{2}})o(m^{\frac{1}{2}}) = o_p(1)$. \\
Expanding \eqref{eq:A1},
\begin{align*}
0 = f(\lambda) &= \frac{1}{n} \sum_l \theta_l^{-1} g(x_l, \beta) \big[ 1 - \gamma_l + \frac{\gamma_l^2}{1+\gamma_l} \big] \\
&= \frac{1}{n} \sum_l \theta_l^{-1} g(x_l, \beta) - \frac{1}{n} \sum_l \theta_l^{-1} g(x_l, \beta)\cdot \gamma  + \frac{1}{n} \sum_l \theta_l^{-1} g(x_l, \beta) \frac{\gamma_l^2}{1+\gamma_l} \\
&= E(\theta_l^{-1} g(x_l, \beta)) - S \lambda + \frac{1}{n} \sum_l \theta_l^{-1} g(x_l, \beta) \frac{\gamma_l^2}{1+\gamma_l}. \tag{A.2}\label{eq:A2}
\end{align*}
The last equality is since $ \frac{1}{n} \sum_l \theta_l^{-1} g(x_l, \beta)\cdot \gamma = \frac{1}{n} \sum_l \theta_l^{-1} g(x_l, \beta) \theta_l^{-1} g^\prime(x_l, \beta)  \lambda= S \lambda.$\\
By substituting $\gamma_l$, we have the final term of \eqref{eq:A2};
\begin{align*}
\frac{1}{n} \sum_l \parallel \theta_l^{-1} g(x_l, \beta)\parallel^3 \parallel \lambda \parallel ^2 |1+\gamma_l|^{-1} = o_p(m^{\frac{1}{2}}) O_p(m^{-1}) o_p(1) = o_p(m^{-\frac{1}{2}}).
\end{align*}
Therefore,
\begin{align*}
0 &=  E(\theta_l^{-1} g(x_l, \beta)) - S \lambda + o_p(m^{-\frac{1}{2}})\\
&\Rightarrow \lambda = S^{-1} E(\theta_l^{-1} g(x_l, \beta)) + o_p(m^{-\frac{1}{2}})\\
&\Rightarrow \lambda = \Big[ \frac{1}{n} \sum_{l=1}^n \theta_l^{-2} g(x_l, \beta) g^\prime(x_l, \beta) \Big]^{-1} \Big[ \frac{1}{n} \sum_{l=1}^n  \theta_l^{-1} g(x_l, \beta)\Big] + o_p(m^{-\frac{1}{2}}).
\tag{A.3}\label{eq:A3}
\end{align*}
Now, denote $V_n(\beta)=\frac{1}{n} \sum_{l=1}^n \theta_l^{-2} g(x_l, \beta)g^\prime(x_l, \beta)$, $\bar{g}(\beta) = \frac{1}{n}\sum_{l=1}^n \theta_l^{-1} g(x_l, \beta)$, and $\varepsilon=\epsilon_k o_p(m^{-\frac{1}{2}})$. So \eqref{eq:A2} can be rewritten as,
$$\lambda(\beta) = V_n(\beta)^{-1} \bar{g}(\beta) + \varepsilon.$$
Since $\gamma_l = \lambda^\prime(\beta)\theta_l^{-1} g(x_l, \beta),$ so $\sum_{l=1}^n |r_l|^3 = o_p(1)$.\\
Let $a_m$ be any constant sequence such that $a_m \rightarrow \infty$, and $a_m m^{-\frac{1}{2}}\rightarrow 0$. Denote the ball $B(\beta_0, a_m) = \{\beta|\parallel \beta-\beta_0 \parallel \leq a_m m^{-\frac{1}{2}}\}$ and the surface of the ball $\partial B(\beta_0, a_m) = \{\beta|\parallel \beta-\beta_0 \parallel = \phi a_m m^{-\frac{1}{2}}, \parallel\phi \parallel=1\}$. For any $\beta\in\partial B(\beta_0, a_m)$, we have
\begin{align*}
V_n(\beta) &= \frac{1}{n} \sum_{l=1}^n \theta_l^{-2} g(x_l, \beta)g^\prime(x_l, \beta)\\
&= \frac{n}{k} \frac{1}{k}\sum_{l=1}^k g_1(x_l, \beta_0)g_1^\prime(x_l, \beta_0) + \frac{n}{n-k} \frac{1}{n-k} \sum_{l=k+1}^n g_2(x_l, \beta_0)g_2^\prime(x_l, \beta_0) + o_p(\epsilon_k^{-1})\\
&= \frac{n}{k} E g_1(x_l, \beta_0)g_1^\prime(x_l, \beta_0) + \frac{n}{n-k} E g_2(x_l, \beta_0)g_2^\prime(x_l, \beta_0) + o_p(\epsilon_k^{-1})\\
&\leq \epsilon_k^{-1}\big[ E g_1(x_l, \beta_0)g_1^\prime(x_l, \beta_0) + E g_2(x_l, \beta_0)g_2^\prime(x_l, \beta_0)\big] + o_p(\epsilon_k^{-1}),
\end{align*}
and
\begin{align*}
\bar{g}(\beta_0) &= \frac{1}{n} \sum_{l=1}^n \theta_l^{-1} g(x_l, \beta) \\
&= \frac{1}{k} \sum_{l=1}^k g_1(x_l, \beta_0) + \frac{1}{n-k} \sum_{l=k+1}^n g_2(x_l, \beta_0)\\
&= \frac{1}{k} o_p(k^{\frac{1}{2}}) + \frac{1}{n-k} o_p((n-k)^{\frac{1}{2}})\\
&= o_p(k^{-\frac{1}{2}}) + o_p((n-k)^{-\frac{1}{2}})\\ %use the minimum order out of the two here.
&= o_p(m^{-\frac{1}{2}}).
\end{align*}
\noindent By the Taylor expansion, for any $\beta\in\partial B(\beta_0, a_m)$, we have
\begin{align*}
l_E(\beta) &= \sum_l \lambda^\prime(\beta)\theta_l^{-1} g(x_l, \beta) - \frac{1}{2} \sum_l \Big[ \lambda^\prime(\beta)\theta_l^{-1} g(x_l, \beta) \Big] ^2 + o_p(1).
\tag{A.4}\label{eq:A4}
\end{align*}
The first term of \eqref{eq:A4} is;
\begin{align*}
\sum_l \lambda^\prime(\beta)\theta_l^{-1} g(x_l, \beta) &= \Big[ \frac{1}{n} \sum_{l=1}^n  \theta_l^{-1} g(x_l, \beta)\Big]^\prime \Big[ \frac{1}{n} \sum_{l=1}^n \theta_l^{-2} g(x_l, \beta) g^\prime(x_l, \beta) \Big]^{-1} \Big[ \frac{1}{n} \sum_{l=1}^n  \theta_l^{-1} g(x_l, \beta)\Big]\\
 &+ o_p(1). \tag{A.4.1}\label{eq:A4_1}
\end{align*}
The second term of \eqref{eq:A4} is:
\begin{align*}
 &\frac{1}{2} \sum_l \Big[ \lambda^\prime(\beta)\theta_l^{-1} g(x_l, \beta) \Big] ^2\\
  &=  \frac{1}{2} \sum_l \lambda^\prime(\beta)\theta_l^{-2} g(x_l, \beta)g^\prime(x_l, \beta)\\
 &= \frac{n}{2}\Big[ \frac{1}{n} \sum_{l=1}^n  \theta_l^{-1} g(x_l, \beta)\Big]^\prime \Big[ \frac{1}{n} \sum_{l=1}^n \theta_l^{-2} g(x_l, \beta) g^\prime(x_l, \beta) \Big]^{-1}\\
 &\Big[ \frac{1}{n} \sum_{l=1}^n \theta_l^{-2} g(x_l, \beta) g^\prime(x_l, \beta) \Big] \Big[ \frac{1}{n} \sum_{l=1}^n \theta_l^{-2} g(x_l, \beta) g^\prime(x_l, \beta) \Big]^{-1}\Big[\frac{1}{n}\sum_{l=1}^n  \theta_l^{-1} g(x_l, \beta)\Big]  + o_p(1)\\
 &= \frac{n}{2}\Big[ \frac{1}{n} \sum_{l=1}^n  \theta_l^{-1} g(x_l, \beta)\Big]^\prime \Big[ \frac{1}{n} \sum_{l=1}^n \theta_l^{-2} g(x_l, \beta) g^\prime(x_l, \beta) \Big]^{-1} \Big[\frac{1}{n}\sum_{l=1}^n  \theta_l^{-1} g(x_l, \beta)\Big]+ o_p(1). \tag{A.4.2}\label{eq:A4_2}
\end{align*}
Now,
\begin{align*}
&\eqref{eq:A4_1} - \eqref{eq:A4_2}\\
&=\frac{n}{2} \Big( \frac{1}{n}\sum_l \theta_l^{-1} g(x_l, \beta) \Big) ^\prime \Big( \frac{1}{n}\sum_l \theta_l^{-2} g(x_l, \beta)g^\prime(x_l, \beta) \Big) ^{-1} \Big( \frac{1}{n}\sum_l \theta_l^{-1} g(x_l, \beta) \Big) + o_p(1).
\end{align*}
So we can rewrite \eqref{eq:A4} as,
\begin{align*}
l_E(\beta)&= \frac{n}{2} \Big( \frac{1}{n}\sum_l \theta_l^{-1} g(x_l, \beta) \Big) ^\prime \Big( \frac{1}{n}\sum_l \theta_l^{-2} g(x_l, \beta)g^\prime(x_l, \beta) \Big) ^{-1} \Big( \frac{1}{n}\sum_l \theta_l^{-1} g(x_l, \beta) \Big) + o_p(1)\\
&= \frac{n}{2} \bar{g}^\prime (\beta) (V_n(\beta))^{-1}\bar{g}(\beta) + o_p(1) \\
&= \frac{n}{2} \Bigg\{\bar{g}(\beta_0) + \frac{1}{n} \sum_l \theta_l^{-1} \frac{\partial g(x_l, \beta_0)}{\partial \beta^\prime}\phi a_m m^{-\frac{1}{2}} + O\Big[(a_m m^{-\frac{1}{2}})^2\Big] \Bigg\}^\prime \times \Big(V_n(\beta)\Big)^{-1} \times\\
& \Bigg\{\bar{g}(\beta_0) + \frac{1}{n} \sum_l \theta_l^{-1} \frac{\partial g(x_l, \beta_0)}{\partial \beta^\prime}\phi a_m m^{-\frac{1}{2}} + O\Big[(a_m m^{-\frac{1}{2}})^2\Big] \Bigg\} + o_p(1) \tag{By Taylor expansion of each term.}\\
& \geq \frac{n \epsilon_k}{2} \Bigg\{\bar{g}(\beta_0) + \frac{1}{n} \sum_l \theta_l^{-1} \frac{\partial g(x_l, \beta_0)}{\partial \beta^\prime}\phi a_m m^{-\frac{1}{2}} + O\Big[(a_m m^{-\frac{1}{2}})^2\Big] \Bigg\}^\prime \times \Big(V_n(\beta)\Big)^{-1} \times\\
& \Bigg\{\bar{g}(\beta_0) + \frac{1}{n} \sum_l \theta_l^{-1} \frac{\partial g(x_l, \beta_0)}{\partial \beta^\prime}\phi a_m m^{-\frac{1}{2}} + O\Big[(a_m m^{-\frac{1}{2}})^2\Big] \Bigg\} + o_p(1).
%& \rightarrow \infty
%&= \frac{n}{2} \Bigg\{\bar{g}(\beta_0) + \frac{1}{n} \sum_l \theta_l^{-1} \frac{\partial g(x_l, \beta_0)}{\partial \beta^\prime}\phi a_m m^{-\frac{1}{2}}  \Bigg\}^\prime \Big(V_n(\beta)\Big)^{-1} \Bigg\{\bar{g}(\beta_0) + \frac{1}{n} \sum_l \theta_l^{-1} \frac{\partial g(x_l, \beta_0)}{\partial \beta^\prime}\phi a_m m^{-\frac{1}{2}} \Bigg\} + o_p(1) \tag{By Taylor expansion of each term.}\\
%&= \frac{n}{2} \Bigg[O\big(m^{-\frac{1}{2}} (\log\log m^{\frac{1}{2}})\big) + E\Big[ \frac{\partial g(x_l, \beta_0)}{\partial \beta^\prime} \Big] \phi a_m m^{-\frac{1}{2}}\Bigg]^\prime \bigg[E\Big(g(x_l, \beta)g^\prime(x_l, \beta) \Big) \bigg]^{-1} \times \\
%& \Bigg[O\big(m^{-\frac{1}{2}} (\log\log m^{\frac{1}{2}})\big) + E\Big[ \frac{\partial g(x_l, \beta_0)}{\partial \beta^\prime} \Big] \phi a_m m^{-\frac{1}{2}}\Bigg] + o_p(1) \tag{Since, $\frac{1}{n} \sum_l \theta_l^{-1} \frac{\partial g(x_l, \beta_0)}{\partial \beta^\prime}=E\Big[ \frac{\partial g_1(x_l, \beta_0)}{\partial \beta^\prime} \Big]$, $\bar{g}(\beta_0)=O\big(m^{-\frac{1}{2}} (\log\log m^{\frac{1}{2}})\big)$}
%
\end{align*}
As $n\rightarrow\infty$, $l_E(\beta) \rightarrow \infty$.

\noindent Similarly,
$$l_E(\beta_0) = \frac{n}{2} \bar{g}^\prime(\beta_0)V_n(\beta_0)^{-1} \bar{g}(\beta_0) + o_p(1),$$
$$V_n(\beta_0) = \frac{n}{k} E g_1(x_l, \beta_0)g_1^\prime(x_l, \beta_0) + \frac{n}{n-k} E g_2(x_l, \beta_0)g_2^\prime(x_l, \beta_0) + o_p(\epsilon_k^{-1}).$$
Thus, $l_E(\beta_0) = O_p(1)$ implies that for any $\beta\in\partial B(\beta_0, a_m)$, $l_E(\beta)$ can not arrive its minimum value with the probability approaching to 1. Since $l_E(\beta)$ is a continuous function about $\beta$, as $\beta\in B(\beta_0, a_m)$, $l_E(\beta)$ has a minimum value in the interior of this ball satisfying,
\begin{align*}
0 = \frac{\partial l_E(\beta)}{\partial \beta}\Bigg|_{\beta=\tilde{\beta}}
&= \sum_l \frac{\bigg(\frac{\partial \lambda^\prime(\beta)}{\partial \beta}\bigg)\theta_l^{-1}g(x_l, \beta) + \theta_l^{-1}\bigg(\frac{\partial g(x_l, \beta)}{\partial \beta}\bigg)^\prime \lambda(\beta)}{1 + \lambda^\prime(\beta) \theta_l^{-1}g(x_l, \beta)} \Bigg|_{\beta=\tilde{\beta}}\\
&= \frac{\partial \lambda^\prime(\beta)}{\partial \beta} \sum_l \frac{\theta_l^{-1}g(x_l, \beta)}{1 + \lambda^\prime(\beta) \theta_l^{-1}g(x_l, \beta)} \Bigg|_{\beta=\tilde{\beta}} + \sum_l \frac{\theta_l^{-1}\bigg(\frac{\partial g(x_l, \beta)}{\partial \beta}\bigg)^\prime \lambda(\beta)}{1 + \lambda^\prime(\beta)\theta_l^{-1} g(x_l, \beta)}\\
&= \sum_l \frac{\theta_l^{-1}\bigg(\frac{\partial g(x_l, \beta)}{\partial \beta}\bigg)^\prime \lambda(\beta)}{1 + \lambda^\prime(\beta) \theta_l^{-1}g(x_l, \beta)} \tag{Since $\sum_l \frac{\theta_l^{-1}g(x_l, \beta)}{1 + \lambda^\prime(\beta) \theta_l^{-1}g(x_l, \beta)} \Big|_{\beta=\tilde{\beta}}= Q_{1n}(\tilde{\beta}, \tilde{\lambda}) = 0 $}\\
&= Q_{2n}(\tilde{\beta}, \tilde{\lambda}).
\end{align*}

\noindent Hence, $Q_{1n}(\tilde{\beta}, \tilde{\lambda}) = 0$ and $ Q_{2n}(\tilde{\beta}, \tilde{\lambda}) = 0$. That is, $\parallel \tilde{\beta}-\beta_0 \parallel = O_p(a_m m^{-\frac{1}{2}}).$ But $a_m$ is arbitrary, hence $\parallel \tilde{\beta}-\beta_0 \parallel = O_p(m^{-\frac{1}{2}})$.
\end{proof}
%%%%

\begin{remark}
From Lemma 3 of Airchison and Silvey (1957), the partitioned matrix
\[
 \begin{pmatrix}
  A & B \\
  B^\prime & 0
 \end{pmatrix}
\]
is non-singular. Hence,
\[
 \begin{pmatrix}
  A & B \\
  B^\prime & 0
 \end{pmatrix}^{-1} =
 \begin{pmatrix}
  P & Q \\
  Q^\prime & R
 \end{pmatrix}
\]
where
$$P=A^{-1} -A^{-1}B(B^\prime A^{-1}  B)^{-1} B^\prime A^{-1}, \qquad Q= A^{-1}B(B^\prime A^{-1}  B)^{-1},$$
$$Q^\prime = (B^\prime A^{-1}  B)^{-1}B^\prime A^{-1}, \qquad R=-(B^\prime A^{-1}  B)^{-1}$$
\end{remark}

\begin{remark}
\noindent
\[\text{If }
 \begin{pmatrix}
  A & B \\
  C & D
 \end{pmatrix}
\text{is a $n\times n$ symmetric positive definite matrix, and the partitioned matrices
$A\in \mathbb{R}^{m\times m},$}\]
\text{ $B\in \mathbb{R}^{m\times n-m}$, }
and
$D\in \mathbb{R}^{(n-m)\times (n-m)}$, then
\begin{enumerate}
\item the matrix $(D-CA^{-1}B)$ is symmetric and positive definite,
\item \[
 \begin{pmatrix}
  A & B \\
  C & D
 \end{pmatrix}^{-1} \geq
 \begin{pmatrix}
  A^{-1} & 0 \\
  0 & 0
 \end{pmatrix}.
\]
\end{enumerate}
\end{remark}

\begin{remark}
$\beta^\prime = ((\beta^\prime, \mu^\prime), \delta^\prime)$.
$$\frac{\partial Q_{1n}(\beta, 0)}{\partial \lambda^\prime} = -\frac{1}{n}\sum_l \theta_l^{-2} g(x_l, \beta)g^\prime(x_l, \beta), \quad \frac{\partial Q_{1n}(\beta, 0)}{\partial \beta^\prime} = \frac{1}{n}\sum_l \theta_l^{-1} \frac{\partial g(x_l, \beta)}{\partial \beta^\prime}$$
$$\frac{\partial Q_{2n}(\beta, 0)}{\partial \lambda^\prime} = \frac{1}{n}\sum_l \bigg(\theta_l^{-1} \frac{\partial g(x_l, \beta)}{\partial \beta^\prime}\bigg)^\prime, \quad \frac{\partial Q_{2n}(\beta, 0)}{\partial \beta^\prime} = 0$$
\item \[
 \begin{pmatrix}
  \frac{\partial Q_{1n}}{\partial \lambda^\prime} & \frac{\partial Q_{1n}}{\partial \beta^\prime} \\
  \frac{\partial Q_{2n}}{\partial \lambda^\prime} & 0
 \end{pmatrix} \longrightarrow
 \begin{pmatrix}
  S_{11} & S_{12} \\
  S_{21} & 0
 \end{pmatrix} = S(\beta) \equiv S
\]
where \\
$S_{11}(\beta) = -\theta_l^{-1} E\big[g_1(x_l, \beta)g_1^\prime(x_l, \beta) \big] - (1-\theta_l)^{-1} E\big[g_2(x_l, \beta)g_2^\prime(x_l, \beta)\big],$\\
$S_{12}(\beta) = \theta^{-1}_lE\Big[ \frac{\partial g_1(x_l, \beta_0)}{\partial \beta^\prime} \Big] + (1-\theta_l)^{-1}E\Big[ \frac{\partial g_2(x_l, \beta_0)}{\partial \beta^\prime} \Big],$\\
$S_{21}(\beta) = S_{12}^\prime(\beta),$\\
$S_{12,i}(\beta) = \theta^{-1}_lE\Big[ \frac{\partial g_1(x_l, \beta_0)}{\partial \beta_i^\prime} \Big] + (1-\theta_l)^{-1}E\Big[ \frac{\partial g_2(x_l, \beta_0)}{\partial \beta_i^\prime} \Big]$, $i=1, 2$.\\
\end{remark}
\noindent By, Remark 1,
\[
S^{-1} =
\begin{pmatrix}
  S_{11} & S_{12} \\
  S_{21} & 0
 \end{pmatrix} ^{-1} =
 \begin{pmatrix}
  P & Q \\
  Q^\prime & R
 \end{pmatrix}
\]
where
$$P = S_{11}^{-1} - S_{11}^{-1} S_{12}(S_{21}S_{11}^{-1}S_{12})^{-1}S_{21}S_{11}^{-1} = S_{11}^{-1} + S_{11}^{-1} S_{12} \Sigma S_{21}S_{11}^{-1}; \quad \Sigma = (S_{21}(-S_{11}^{-1})S_{12})^{-1},$$
$$Q = - S_{11}^{-1} S_{12}(S_{21}S_{11}^{-1}S_{12})^{-1} = - S_{11}^{-1} S_{12} \Sigma, \qquad Q^\prime = -\Sigma S_{21}S_{11}^{-1}, \quad R = -(S_{21}S_{11}^{-1}S_{12})^{-1} = \Sigma.$$

%%%%
\begin{lemma}
Under the conditions in Lemma 1 and $H_0$, as $n\rightarrow \infty$ we have
$$\sqrt{n}\Sigma^{-\frac{1}{2}}(\tilde{\beta}-\beta_0) \rightarrow N(0, I_{2p+q}),$$
where $\Sigma = [S_{21}(-S_{11})^{-1}S_{12}]^{-1}$.
\end{lemma}

\begin{proof}
Expanding $ Q_{1n}(\tilde{\beta}, \tilde{\lambda})$ and $ Q_{2n}(\tilde{\beta}, \tilde{\lambda})$ at $(\theta_0,0)$, by the conditions of the $H_0$ and Lemma 1, we have,
\begin{align*}
0 &= Q_{1n}(\tilde{\beta}, \tilde{\lambda})\\
&= Q_{1n}(\beta_0, 0) + \frac{\partial Q_{1n}(\beta_0, 0)}{\partial \beta^\prime} (\tilde{\beta} - \beta_0) + \frac{\partial Q_{1n}(\beta_0, 0)}{\partial \lambda^\prime} (\tilde{\lambda} - 0) + O_p(m^{-1}),
\end{align*}
\begin{align*}
0 &= Q_{2n}(\tilde{\beta}, \tilde{\lambda})\\
&= Q_{2n}(\beta_0, 0) + \frac{\partial Q_{2n}(\beta_0, 0)}{\partial \beta^\prime} (\tilde{\beta} - \beta_0) + \frac{\partial Q_{2n}(\beta_0, 0)}{\partial \lambda^\prime} (\tilde{\lambda} - 0) + O_p(m^{-1}),
\end{align*}

\[
 \begin{pmatrix}
 -Q_{1n}(\beta_0, 0) + O_p(m^{-1}) \\
  \epsilon_k O_p(m^{-1})
 \end{pmatrix} =
 \begin{pmatrix}
 \frac{\partial Q_{1n}}{\partial \lambda^\prime} & \frac{\partial Q_{1n}}{\partial \beta^\prime} \\
  \frac{\partial Q_{2n}}{\partial \lambda^\prime} & 0
 \end{pmatrix}
  \begin{pmatrix}
  \tilde{\lambda} \\
  \tilde{\beta} -\beta_0
 \end{pmatrix}.
\]
By LLN,
\[
\begin{pmatrix}
  \tilde{\lambda} \\
  \tilde{\beta} -\beta_0
 \end{pmatrix} \longrightarrow S^{-1}(\beta_0)
 \begin{pmatrix}
  -Q_{1n}(\beta_0, 0) + O_p(m^{-1}) \\
  \epsilon_k O_p(m^{-1}).
 \end{pmatrix}
\]
By Remark 1,
\[ \tilde{\beta} -\beta_0 = (0 ~I) S^{-1}
\begin{pmatrix}
  -Q_{1n}(\beta_0, 0) + O_p(m^{-1}) \\
  \epsilon_k O_p(m^{-1})
 \end{pmatrix}.
\]
Therefore,
\[
\begin{pmatrix}
  \tilde{\lambda} \\
  \tilde{\beta} -\beta_0
 \end{pmatrix} \longrightarrow
 \begin{pmatrix}
  S_{11}^{-1} + S_{11}^{-1} S_{12} & - S_{11}^{-1} S_{12} \Sigma \\
  -\Sigma S_{21}S_{11}^{-1} & \Sigma
 \end{pmatrix}
 \begin{pmatrix}
  -Q_{1n}(\beta_0, 0) + O_p(m^{-1}) \\
  \epsilon_k O_p(m^{-1})
 \end{pmatrix}.
\]

\begin{align*}
\tilde{\beta} -\beta_0 &\rightarrow -\Sigma S_{21}S_{11}^{-1} \bigg( -Q_{1n}(\beta_0, 0) + O_p(m^{-1}) \bigg) + \Sigma \epsilon_k O_p(m^{-1})\\
&= (S_{21}S_{11}^{-1}S_{12})^{-1} S_{21}S_{11}^{-1}Q_{1n}(\beta_0, 0) - \Sigma S_{21}S_{11}^{-1} O_p(m^{-1}) + \Sigma \epsilon_k O_p(m^{-1})\\
&= \frac{1}{\sqrt{n}} (S_{21}S_{11}^{-1}S_{12})^{-1} S_{21}(-S_{11})^{-1/2}(-S_{11})^{-1/2}\sqrt{n}Q_{1n}(\beta_0, 0) + \epsilon_k O_p(m^{-\frac{1}{2}})
\end{align*}
Since $(-S_{11})^{-1/2}\sqrt{n}Q_{1n}(\beta_0, 0) \rightarrow N(0, I_{2(p+q)})$,
%$\sqrt{n}\Sigma^{-1/2}(\tilde{\beta} - \beta_0) \rightarrow N(0, I_{2p+q})$.
$\sqrt{n}S_{21}S_{11}^{-1}(\tilde{\beta} - \beta_0) \rightarrow N(0, I_{2p+q})$.
\end{proof}
%\vspace{0.2in}
%%%%
\begin{lemma}
$$-2\log \Lambda_k = 2l_E(\tilde{\beta}_1^0, 0) - 2l_E(\tilde{\beta}_1^0, \tilde{\beta}_2^0),$$
where $\tilde{\beta}_1^0$ minimizes $l_E(\beta, 0)$ with respect to $\beta_1$ under $H_0$,
$$-2\log \Lambda_k  = \Big[ (-S_{11})^{-1/2}\sqrt{n}Q_{1n}(\beta_0, 0) \Big] ^\prime \Delta \Big[ (-S_{11})^{-1/2}\sqrt{n}Q_{1n}(\beta_0, 0) \Big] + O_p(m^{-\frac{1}{2}}) $$
where $$\Delta = (-S_{11})^{-1/2} \Big \{S_{12} [S_{21}(-S_{11})^{-1}S_{12}]^{-1} S_{21} - S_{12,1} [S_{21,1}(-S_{11})^{-1}S_{12,1}]^{-1} S_{21,1} \Big \}(-S_{11})^{-1/2} \geq 0.$$
\end{lemma}

\begin{proof}
Similar to Qin and Lawless (1994), we can derive,
$$l_E(\tilde{\beta}_1^0, \tilde{\beta}_2^0) = -\frac{n}{2} Q_{1n}^\prime (\beta_0, 0) B Q_{1n} (\beta_0, 0) +  O_p(m^{-\frac{1}{2}}), $$
where $B = S_{11}^{-1} + S_{11}^{-1} S_{12} \Sigma S_{21}S_{11}^{-1}$, and
$$l_E(\tilde{\beta}_1^0, 0) = -\frac{n}{2} Q_{1n}^\prime (\beta_0, 0) A Q_{1n} (\beta_0, 0) +  O_p(m^{-\frac{1}{2}}), $$
where $A = S_{11}^{-1} + S_{11}^{-1} S_{12,1} (S_{21,1}S_{11}^{-1}S_{12,1})^{-1} S_{21,1}S_{11}^{-1}$.
Then,
\begin{align*}
2 \Big[ l_E(\tilde{\beta}_1^0, 0) - l_E(\tilde{\beta}_1^0, \tilde{\beta}_2^0)\Big]
&=\bigg[ -Q_{1n}^\prime (\beta_0, 0) A Q_{1n} (\beta_0, 0) +  O_p(m^{-\frac{1}{2}})\bigg] + \\
&\bigg[n Q_{1n}^\prime (\beta_0, 0) B Q_{1n} (\beta_0, 0) +  O_p(m^{-\frac{1}{2}})\bigg]\\
&= n Q_{1n}^\prime (\beta_0, 0) (B-A) Q_{1n} (\beta_0, 0) +  O_p(m^{-\frac{1}{2}})\\
&= n Q_{1n}^\prime (\beta_0, 0) S_{11}^{-1} \Big[S_{12} \Sigma S_{21} - S_{12,1} \Sigma^* S_{12,2}\Big] S_{11}^{-1} Q_{1n} (\beta_0, 0) +  O_p(m^{-\frac{1}{2}})
\tag{$B-A = S_{11}^{-1} + S_{11}^{-1} S_{12} \Sigma S_{21}S_{11}^{-1} - S_{11}^{-1} - S_{11}^{-1} S_{12,1} (S_{21,1}S_{11}^{-1}S_{12,1})^{-1} S_{21,1}S_{11}^{-1}$.}\\
\tag{So, $\Sigma^*=(S_{21,1}S_{11}^{-1}S_{12,1})^{-1}$}\\
&= \Big[ (-S_{11})^{-1/2}\sqrt{n}Q_{1n}(\beta_0, 0)\Big]^\prime (-S_{11})^{-1/2} \Big[S_{12} \Sigma S_{21} - S_{12,1} \Sigma^* S_{12,2}\Big]\\
&(-S_{11})^{-1/2} \Big[(-S_{11})^{-1/2}\sqrt{n}Q_{1n}(\beta_0, 0)\Big] + O_p(m^{-\frac{1}{2}}).
\end{align*}

\noindent Take $\Delta = (-S_{11})^{-1/2} \Big[S_{12} \Sigma S_{21} - S_{12,1} \Sigma^* S_{12,2}\Big] (-S_{11})^{-1/2}$. Now,
\begin{align*}
\Delta &= (-S_{11})^{-1/2} \Big[S_{12} \Big(S_{21}(-S_{11}^{-1})S_{12}\Big)^{-1} S_{21} - S_{12,1} \Big(S_{21,1}S_{11}^{-1}S_{12,1}\Big)^{-1} S_{12,2}\Big] (-S_{11})^{-1/2}\\
&= (-S_{11})^{-1/2} (S_{12,1},~ S_{12,2}) \Bigg\{[S_{21}(-S_{11})^{-1}S_{12}]^{-1} - \begin{pmatrix}
    \Big(S_{21,1}S_{11}^{-1}S_{12,1}\Big)^{-1} & 0 \\
    0 & 0
\end{pmatrix} \Bigg\}\\
&\times \begin{pmatrix}
    S_{21,1}\\
    S_{21,2}
\end{pmatrix}
(-S_{11})^{-1/2} \\
&\geq 0. \tag{By Remark 2}
\end{align*}
\end{proof}
%%%%
\begin{lemma}
Under the conditions of Theorem 1 and the null hypothesis, denote
$U_{n_k} = \Big\{\frac{k}{n}: \frac{T}{n} \leq (1-\frac{T}{n})\Big\}$, for all $\delta>0$, we can find $C = C(\delta)$, $T=T(\delta)$ and $N=N(\delta)$ such that
$$P\bigg( \smash{\displaystyle\max_{\frac{k}{n} \in U_{n_k}}} \Big(\frac{m}{\log\log m}\Big)^{1/2} \parallel \frac{\tilde{\lambda}}{\epsilon_k}\parallel > C \bigg) \leq \delta, \quad
P\bigg( n^{-1/2} \smash{\displaystyle\max_{\frac{k}{n} \in U_{n_k}}} m \parallel \frac{\tilde{\lambda}}{\epsilon_k}\parallel > C \bigg) \leq \delta, $$
$$P\bigg( \smash{\displaystyle\max_{\frac{k}{n} \in U_{n_k}}} \Big(\frac{m}{\log\log m}\Big)^{1/2} \parallel \tilde{\theta} - \theta_0 \parallel > C \bigg) \leq \delta, \quad
P\bigg( n^{-1/2} \smash{\displaystyle\max_{\frac{k}{n} \in U_{n_k}}} m \parallel \tilde{\theta} - \theta_0 \parallel > C \bigg) \leq \delta.$$
\end{lemma}
\begin{proof}
The proof is similar to Lemma 1.2.2 of Cs\"{o}rg\'{o} and Horv\'{a}th (1997).
\end{proof}

\begin{lemma}
Under the conditions of Theorem 1 and $H_0$, for all $0\leq \alpha < \frac{1}{2}$ we have:
$$n^\alpha  \smash{\displaystyle\max_{\frac{k}{n} \in U_{n_k}}} \big[ \theta_k (1-\theta_k) \big]^\alpha | -2\log \Lambda - R_k| = O_p(1),$$
$$ \smash{\displaystyle\max_{\frac{k}{n} \in U_{n_k}}} \big[ \theta_k (1-\theta_k) \big]^\alpha | -2\log \Lambda - R_k| = O_p(n^{-\frac{1}{2}}(\log\log n)^{\frac{3}{2}}),$$
where $\Theta_{nk}=\{k:\delta_1 \leq k \leq n-\delta_2\}$
\end{lemma}
\noindent\textbf{Proof of Theorem 1:}
\begin{proof}
The proof of Theorem 1 is similar to the proof of Theorem 1.3.1 (Theorem A.3.4) of Cs\"{o}rg\'{o} and Horv\'{a}th (1997) which derives the null distribution of the trimmed test statistic. %Also Part seven result has been used.
\end{proof}

\noindent\textbf{Proof of Theorem 2:}
\begin{proof}
The ELR test statistic is,
$$-2 \log \Lambda_k = Z_{H_0,k_0} - Z_{H_1,k_0}.$$
Under $H_1$, $Z_{H_1,k_0}$ also follows an asymptotic $\chi^2$ distribution. Therefore, $Z_{H_1,k_0} = O_p(1)$. We only need to prove that $P(Z_{H_0,k_0}>cn) \rightarrow 1$ for a positive constant $c$ under $H_1$. For any fixed $\varepsilon$, we can obtain
\begin{align*}
\frac{1}{2n}Z_{H_0,k_0} &= \smash{\displaystyle\sup_{\lambda}}\frac{1}{n}\sum_{l=1}^n \log \Big[ 1 + \theta_l^{-1} \lambda^\prime g(x_l, \varepsilon) \Big]\\
&= \smash{\displaystyle\sup_{\lambda_1}}\frac{1}{n}\sum_{l=1}^{k_0} \log \Big[ 1 + \theta_{k_0}^{-1} \lambda_1^\prime g_1(x_l, \varepsilon) \Big] + \smash{\displaystyle\sup_{\lambda_2}}\frac{1}{n}\sum_{l=k_0+1}^{n} \log \Big[ 1 + (1- \theta_{k_0})^{-1} \lambda_2^\prime g_2(x_l, \varepsilon) \Big]\\
%&\xrightarrow{\text{a.s.}}\smash{\displaystyle\sup_{\lambda_1}} \theta_0 E \log \Big( 1 + \theta_0^{-1} \lambda_1^\prime xe \Big) + \smash{\displaystyle\sup_{\lambda_2}} (1-\theta_0) E \big[ \log \Big( 1 + (1-\theta_0)^{-1} \lambda_2^\prime x(x^\prime \delta + e)\Big)\big] \tag{since $h_1(x_l,\varepsilon)=z_l e_l$ for $1\leq l \leq k_0$ and $h_2(x_l,\varepsilon) = z_l(x_l^\prime \delta + e_l)$ for $k_0+1\leq l \leq n$.}
&\xrightarrow{\text{a.s.}}\smash{\displaystyle\sup_{\lambda_1}} \theta_0 E \log \Big( 1 + \theta_0^{-1} \lambda_1^\prime g_1(x_l, \varepsilon) \Big) + \smash{\displaystyle\sup_{\lambda_2}} (1-\theta_0) E \big[ \log \Big( 1 + (1-\theta_0)^{-1} \lambda_2^\prime g_2(x_l, \varepsilon)\Big)\big]
\end{align*}
By Jensen's inequality,
\begin{align*}
E \log \Big( 1 &+ \theta_0^{-1} \lambda_1^\prime g_1(x_l, \varepsilon) \Big)  \leq \log \bigg[E \Big( 1 + \theta_0^{-1} \lambda_1^\prime g_1(x_l, \varepsilon) \Big) \bigg] = 0\\
&\Longrightarrow \smash{\displaystyle\sup_{\lambda_1}} \theta_0 E \log \Big( 1 + \theta_0^{-1} \lambda_1^\prime g_1(x_l, \varepsilon) \Big) = 0.
\end{align*}
Thus,
\begin{align*}
\frac{1}{2n}Z_{H_0,k_0} &\xrightarrow{\text{a.s.}} \smash{\displaystyle\sup_{\lambda_2}} (1-\theta_0) E \big[ \log \Big( 1 + (1-\theta_0)^{-1} \lambda_2^\prime g_2(x_l, \varepsilon)\Big)\big] \\
&\leq (1-\theta_0) c_0
\end{align*}
Hence, $P(Z_{H_0, k_0} \geq (1-\theta_0c_0) \rightarrow 1$.
Thus, the proof.
\end{proof}

\noindent\textbf{Proof of Theorem 3:}
\begin{proof}
To prove: For arbitrary small $\frac{\theta_0}{2} > \eta$, $|\frac{k_0-k}{n}|\geq \eta$, $-2\log\Lambda_k$ cannot arrive at its maximum with probability approaching to 1. \\
%By the definition of $\hat{k}$, we have $|\frac{k_0-\hat{k}}{n}|\leq \eta$ with probability approaching to 1. Since $\eta$ is arbitrary, thus the proof.
Without loss of generality, suppose $k<k_0$ and $\frac{k_0-k}{n}\geq \eta$. Then we have,
$$-2 \log \Lambda_{k_0} - (-2 \log \Lambda_k) = (Z_{H_0,k_0} - Z_{H_1,k_0}) - (Z_{H_0,k} - Z_{H_1,k}).$$
Since $Z_{H_1,k_0} = O_p(1)$
\begin{align*}
\frac{1}{2n}(Z_{H_0,k_0} & -Z_{H_0,k} + Z_{H_1,k}) \\
&= \smash{\displaystyle\sup_{\lambda_1}}\frac{1}{n}\sum_{l=1}^{k_0} \log \Big[ 1 + \theta_{k_0}^{-1} \lambda_1^\prime g_1(x_l, \beta_0) \Big] + \smash{\displaystyle\sup_{\lambda_2}}\frac{1}{n}\sum_{l=k_0+1}^{n} \log \Big[ 1 + (1- \theta_{k_0})^{-1} \lambda_2^\prime g_2(x_l, \beta_0) \Big]\\
&- \smash{\displaystyle\sup_{\lambda_1}}\frac{1}{n}\sum_{l=1}^{k} \log \Big[ 1 + \theta_{k}^{-1} \lambda_1^\prime g_1(x_l, \beta_0) \Big] + \smash{\displaystyle\sup_{\lambda_2}}\frac{1}{n}\sum_{l=k+1}^{n} \log \Big[ 1 + (1- \theta_{k})^{-1} \lambda_2^\prime g_2(x_l, \beta_0) \Big]\\
&+ \smash{\displaystyle\sup_{\lambda_1}}\frac{1}{n}\sum_{l=1}^{k} \log \Big[ 1 + \theta_{k}^{-1} \lambda_1^\prime g_1(x_l, \beta_0) \Big] + \smash{\displaystyle\sup_{\lambda_2}}\frac{1}{n}\sum_{l=k+1}^{n} \log \Big[ 1 + (1- \theta_{k})^{-1} \lambda_2^\prime g_2(x_l, \beta_1) \Big]\\
&\geq  \smash{\displaystyle\sup_{\lambda_2}}\frac{1}{n}\sum_{l=k_0+1}^{n} \log \Big[ 1 + (1- \theta_{k_0})^{-1} \lambda_2^\prime g_2(x_l, \beta_0) \Big] - \smash{\displaystyle\sup_{\lambda_2}}\frac{1}{n}\sum_{l=k+1}^{n} \log \Big[ 1 + (1- \theta_{k})^{-1} \lambda_2^\prime g_2(x_l, \beta_0) \Big]\\
&=  \smash{\displaystyle\sup_{\lambda_2}}\frac{1}{n}\sum_{l=k_0+1}^{n} \log \Big[ 1 + (1- \theta_{k_0})^{-1} \lambda_2^\prime g_2(x_l, \beta_0) \Big] \\
&- \smash{\displaystyle\sup_{\lambda_2}}\frac{1}{n}\sum_{l=k+1}^{n} \log \Big[ 1 + \rho_k(1- \theta_{k_0})^{-1} \lambda_2^\prime g_2(x_l, \beta_0) \Big] \tag{$\rho_k=\frac{n-k_0}{n-k}$}\\
&=  \smash{\displaystyle\sup_{\lambda_2}}\frac{1}{n}\sum_{l=k_0+1}^{n} \log \Big[ 1 + (1- \theta_{k_0})^{-1} \lambda_2^\prime g_2(x_l, \beta_0) \Big] \\
&- \smash{\displaystyle\sup_{\lambda_2}}\frac{1}{n}\sum_{l=k+1}^{n} \log \Big[ 1 + (1- \theta_{k_0})^{-1} \lambda_2^\prime g_2(x_l, \beta_0) \Big]
\tag{Since, $\frac{n-k_0}{n}\leq \rho_k\leq\frac{n-k_0}{n-k_0+n\eta}$. So, $1-\theta_0\leq \varliminf \rho_k \leq \varlimsup \rho_k \leq \frac{1-\theta_0}{1-\theta_0+\eta}$}\\
&\xrightarrow{\text{a.s.}}\smash{\displaystyle\sup_{\lambda_1}} (1-\theta_0) E \big[ \log \Big( 1 + (1-\theta_0)^{-1} \lambda_2^\prime g_2(x_l, \beta_0)\Big)\big] \\
&- \smash{\displaystyle\sup_{\lambda_1}} \big\{\frac{k_0-k}{n} E \big[ \log \Big( 1 + (1-\theta_0)^{-1} \lambda_2^\prime g_2(x_l, \beta_0)\Big)\big]\\
&+ (1-\theta_0) E \big[ \log \Big( 1 + (1-\theta_0)^{-1} \lambda_2^\prime g_2(x_l, \beta_0)\Big)\big]\big\}\\
&\geq \smash{\displaystyle\sup_{\lambda_1}} (1-\theta_0) E \big[ \log \Big( 1 + (1-\theta_0)^{-1} \lambda_2^\prime g_2(x_l, \beta_0)\Big)\big] \\
&- \smash{\displaystyle\sup_{\lambda_1}} \big\{\eta E \big[ \log \Big( 1 + (1-\theta_0)^{-1} \lambda_2^\prime g_2(x_l, \beta_0)\Big)\big] \\
&+ (1-\theta_0) E \big[ \log \Big( 1 + (1-\theta_0)^{-1} \lambda_2^\prime g_2(x_l, \beta_0)\Big)\big]\big\} \tag{By Jensen's inequality and $\frac{k_0-k}{n}\geq \eta$.}\\
\end{align*}
Assume that $\smash{\displaystyle\sup_{\lambda_1}} \big\{\eta E \big[ \log \Big( 1 + (1-\theta_0)^{-1} \lambda_2^\prime g_2(x_l, \beta_0)\Big)\big] + (1-\theta_0) E \big[ \log \Big( 1 + (1-\theta_0)^{-1} \lambda_2^\prime g_2(x_l, \beta_0)\Big)\big]\big\}$ attains its maximum at $\delta_2^*$. Then we have,
\begin{align*}
\frac{1}{2n}(Z_{H_0,k_0} & -Z_{H_0,k} + Z_{H_1,k}) \\
&\geq
\begin{cases}
\smash{\displaystyle\sup_{\lambda_1}} (1-\theta_0) E \big[ \log \Big( 1 + (1-\theta_0)^{-1} \lambda_2^\prime g_2(x_l, \beta_0)\Big)\big], \text{ if $\delta_2^*=0$},\\
-\eta E \big[ \log \Big( 1 + (1-\theta_0)^{-1} \lambda_2^{*\prime} g_2(x_l, \beta_0)\Big)\big], \text{ if $\delta_2^*\neq 0$}.
\end{cases}
\end{align*}
Therefore, by the condition that for every fixed parameter $\delta=\beta^*-\beta\neq0$, there exists a positive constant $c_0>0$ satisfy that $\infty > \smash{\displaystyle\inf_{\delta\neq 0}} \smash{\displaystyle\sup_{\lambda}} E\log \big[ 1 + \lambda^\prime x(x^\prime \delta + e) \big] \geq c_0 >0$ and Jensen's inequality, there exists a constant $c_0>0$, such that $P\big(\frac{1}{2n}(Z_{H_0,k_0} -Z_{H_0,k} + Z_{H_1,k}) > c_0\big) \rightarrow 1$ as $n\rightarrow \infty$. Thus, we have, $P\big[\big(-2 \log \Lambda_{k_0} - (-2 \log \Lambda_k)\big)>cn\big]\rightarrow 1$, since $Z_{H_1,k_0} = O_p(1)$. So, $-2\log\Lambda_k$ cannot arrive at its maximum with probability approaching to 1. By the definition of $\hat{k}$, we have $|\frac{k_0-\hat{k}}{n}|\leq \eta$ with probability approaching to 1. Since $\eta$ is arbitrary, thus the proof. \\
\end{proof}

\end{document}